\RequirePackage{luatex85}
\documentclass[a4paper,oneside,reqno]{amsart}

\usepackage[english]{babel}
\usepackage[utf8]{inputenc}		
\usepackage{lmodern}
\usepackage[scaled]{helvet}		
\usepackage{courier}			
\usepackage{eulervm}			
\usepackage[bb=boondox,bbscaled=1.05,scr=dutchcal]{mathalfa}	%
\usepackage{dsfont}
\usepackage[T1]{fontenc}
\normalfont

\usepackage[a4paper,margin=2.5cm]{geometry}
\usepackage[dvipsnames]{xcolor}
\usepackage[hypertexnames=false,colorlinks=true,linkcolor=blue,%
citecolor=purple,filecolor=magenta,urlcolor=cyan]{hyperref}                      
\usepackage{amssymb,anyfontsize,enumerate,layout,mathrsfs,mdwlist,stmaryrd,thmtools,tikz}
\usepackage{dsfont} 
\usepackage[textsize=footnotesize]{todonotes}
\presetkeys%
{todonotes}%
{color=Apricot}{}%
\usepackage[style=alphabetic,maxnames=99,maxalphanames=5]{biblatex}
\usepackage{csquotes}
\usepackage{mathtools}

\usepackage[english,noabbrev]{cleveref}
\crefname{lemma}{lemma}{lemmata}
\Crefname{lemma}{Lemma}{Lemmata}

\allowdisplaybreaks 

\theoremstyle{plain}                          
\newtheorem{theorem}{Theorem}[section]
\newtheorem{proposition}[theorem]{Proposition}    
\newtheorem{lemma}[theorem]{Lemma}

\newtheorem{conjecture}[theorem]{Conjecture}
\theoremstyle{definition}
\newtheorem{definition}[theorem]{Definition}
\newtheorem{prop-defin}[theorem]{Proposition-definition} 
\newtheorem{example}[theorem]{Example}
\theoremstyle{remark}
\newtheorem{remark}[theorem]{Remark}

\renewcommand{\theta}{\vartheta}
\renewcommand{\phi}{\varphi}
\renewcommand{\epsilon}{\varepsilon}

\newcommand{\bs}[1]{\ensuremath{\boldsymbol{#1}}}
\newcommand{\mb}[1]{\mathbb{#1}} 

\newcommand{\mc}[1]{\mathcal{#1}}
\newcommand{\ms}[1]{\mathscr{#1}}

\newcommand{\R}{\mb{R}} 
 
\newcommand{\C}{\mb{C}} 
\newcommand{\Z}{\mb{Z}} 
\newcommand{\Q}{\mb{Q}}
\renewcommand{\P}{\mb{P}}

\newcommand{\<}{\langle}
\renewcommand{\>}{\rangle}

\DeclareMathOperator{\csch}{csch}

\DeclareMathOperator*{\Res}{Res}

\begingroup\expandafter\expandafter\expandafter\endgroup
\expandafter\ifx\csname pdfsuppresswarningpagegroup\endcsname\relax
\else
\fi
{}

\begin{document}
	
\title{The Laplace Transform and Quantum Curves}

\author[Q.~Weller]{Quinten Weller}
\email{qgw1@nyu.edu}
\address{New York University, Department of Physics, 726 Broadway, New York, NY 10003, USA}	
	
\thanks{}

\begin{abstract}
A Laplace transform that maps the topological recursion (TR) wavefunction to its $x$-$y$ swap dual is defined. This transform is then applied to the construction of quantum curves. General results are obtained, including a formula for the quantisation of many spectral curves of the form $e^xP_2(e^y) - P_1(e^y) = 0$ where $P_1$ and $P_2$ are coprime polynomials; an important class that contains interesting spectral curves related to mirror symmetry and knot theory that have, heretofore, evaded the general TR-based methods previously used to derive quantum curves. Quantum curves known in the literature are reproduced, and new quantum curves are derived. In particular, the quantum curve for the $T$-equivariant Gromov-Witten theory of $\P(a,b)$ is obtained. 
\end{abstract}
	
\maketitle
	
\tableofcontents


\section{Introduction}
\subsection{Motivation}\label{ss:mot}

Topological recursion (TR) was originally discovered as a recursive solution to the large-$N$ expansion of correlation functions in matrix models \cite{E04,CEO06}. It took in as initial data a so-called spectral curve, which is an algebraic curve $P(x,y) = 0$ with additional structure on it's first homology group;\footnote{Rather than directly specifying homological cycles, this structure is often presented in terms of a choice of symmetric bidifferential denoted by $\omega_{0,2}$.} the spectral curve itself was derived from the matrix model action \cite{E04,CEO06}. In \cite{EO07} it was realised that the topological recursion should be considered as a generic way to associate an infinite tower of $n$-differentials $\{\omega_{g,n}\}_{g,n\geq 0}$ to a given spectral curve. Since the key insight of \cite{EO07} the topological recursion has found applications in many areas of mathematics and physics \cite{E14}; most relevantly for the present work topological recursion has applications in mirror symmetry and Gromov-Witten theory \cite{BKMP08,BM08,GS11,BS12,Z12,DOSS14,NS14,EO15,FLZ17,FLZ20,GKLS22} as well as knot theory \cite{ABM12,BE12,GJKS14,DPSS19}.

In a $1$-matrix model the expectation value of the characteristic polynomial is known to be related to systems of orthogonal polynomials \cite{M90}. Many systems of orthogonal polynomials satisfy differential equations, so in a matrix model it is then natural to expect that this expectation value also satisfies a differential equation. In the abstracted language of the topological recursion this is known as the topological recursion/quantum curve correspondence.

Indeed, this correspondence posits the existence of a quantised form of $P(x,y) = 0$
\begin{equation}
	\hat{P}(\hat{x},\hat{y};\hbar)\psi(x;\hbar) = 0\,,
\end{equation}
where $\hat{x} \coloneq x \cdot$, $\hat{y} \coloneq \hbar d/dx$, $\hbar$ is a quantum parameter (the limit $\hbar \to 0$ should recover the classical curve), and $\psi$ is naturally defined in terms of the $\omega_{g,n}$ in a way that imitates the expectation value of the characteristic polynomial in a matrix model.

The topological recursion/quantum curve correspondence has been rigorously and systematically established for many large classes of curves \cite{BE17,EG19,MO20,EGMO21,MO22,BKW23}. However, the applications of topological recursion to mirror symmetry and knot theory require curves of the form $P(e^x,e^y) = 0$ \cite{BKMP08,BE12} where $P$ is a polynomial: these cases are not covered by any of the above works. Furthermore, this correspondence takes on a special importance for curves related to knot theory as the asymptotics of the wavefunction solution of the quantisation of the $A$-polynomial associated to a knot (here $A=P$) is believed to hold information on the associated knot invariants \cite{MM01,G04,G03,DFM11}.

In this work, the special case where $P(e^x,e^y) = 0$ is linear in its first argument is tackled. The key insight is that, for these spectral curves, the dual wavefunction $\psi^{\vee}$, defined by swapping the roles of $x$ and $y$ in the topological recursion, is particularly simple. Indeed, one can often write down $\hat{P}^{\vee}$ based on the explicit form of $\psi^{\vee}$.

However, $\psi^{\vee}$ is not the object of interest, rather $\psi$ is. To remedy this a notion of Laplace transform on a spectral curve $\mc{S}$ is defined with the property that 
\begin{equation}
	\psi(x;\hbar) = \ms{L}_{\mc{S}}^{-1}\left[\psi^{\vee}(y;\hbar)\right](x)\,.
\end{equation}
This Laplace transform is studied and its key properties are established. In particular it is demonstrated how one can use the Laplace transform to create a dictionary translating dual quantum curves to quantum curves for the original spectral curve of interest. This allows for a general result regarding the quantisation of certain curves of the form $P(e^x,e^y) = 0$, as well as providing a simple method to quantise many other curves.

\subsection{Main results}

The notion of Laplace transform central to the present work is defined in \cref{d:Lap}, it is shown that the Laplace transform of the wavefunction is the dual wavefunction in \cref{t:Lap}, and further key properties of the Laplace transform are established in \cref{l:LapInv,l:LapParts}. A general result for the quantisation of curves of the form $P_2(e^y) e^x - P_1(e^y)  = 0$ is obtained in \cref{t:QC}.

The main application of the Laplace transform studied here is to the curve 
\begin{equation}
	P(x,y) = x - \sum_{l=-b}^{a} \left[ q_le^{ly}+w_l\log(q_le^{ly}) \right] = 0\,,
\end{equation}
where $a,b\in\Z_{\geq 1}$, which is related to the Gromov-Witten theory of $\P(a,b)$ \cite{T18}. From the Laplace transform, the following quantum spectral curve is obtained (\cref{t:GW})
\begin{equation}
	\left[ \hat{x}+\frac{1}{2}\hbar - \sum_{l=-b}^{a} \left[ q_le^{l\hat{y}}+w_l(l\hat{y}+\log q_l) \right] - C \right]\psi(x;\hbar) = 0\,,
\end{equation}
where $C\in\bigoplus_{l}2w_l\pi i\Z$ is a constant depending on the branch choice of the logarithm. Many special cases of this result have been established in the literature \cite{DMNPS17,CG19,A21}. See the discussion succeeding \cref{t:GW} for more details.

Additionally, quantum curves for the spectral curves $P(x,y) = e^x + e^{(f+1)y} - e^{fy} = 0$, $P(x,y) = c^P e^x(e^{y}-c^2)^Q - c^Q e^{Py}(e^{y}-1)^Q = 0$, and $P(x,y) = e^x - ye^{-\phi(y)} = 0$ are obtained in \cref{sss:mc,sss:QP,sss:H}, respectively.

\subsection{Outline}

In \cref{s:back} the necessary background is reviewed. The topological recursion is defined (actually, the so-called LogTR of \cite{ABDKS24} is used), and the notions of spectral curve and quantum curve are reviewed. Furthermore, a crash course in the $x$-$y$ swap duality is provided. 

Next, in \cref{s:Lap} the Laplace transform (or, better, the version of the Laplace transform used in the present work) is defined, and its key properties are discussed and proven. In \cref{ss:gen} the use of the Laplace transform in deriving quantum curves is discussed, and examples of the Laplace transform in action are provided. Furthermore, a general result on the quantisation of curves of the form $P_2(e^y) e^x - P_1(e^y) = 0$ is established. 

In \cref{sss:GWP1}, the Gromov-Witten theory of $\P^1$ is reviewed, and the quantum spectral curve is derived using the machinery built in previous sections. Generalising \cref{sss:GWP1}, \cref{sss:GWPab} discusses the $T$-equivariant Gromov-Witten theory of $\P(a,b)$, and the associated quantum spectral curve is computed.

Finally, the paper is summarised, and directions for future study are discussed in \cref{s:con}.

\subsection{Acknowledgments}

The author would like to thank Sergey Shadrin and Yifan Wang for useful discussions, Vincent Bouchard for pointing out the key article \cite{T18}, and acknowledges that his research is supported in part by the NSERC PGS-D program. The author would especially like to thank Reinier Kramer for his useful comments regarding an earlier draft of the present work.

\section{Background}\label{s:back}
\subsection{Spectral curves and topological recursion}

The central object in the theory of topological recursion is the so-called \emph{spectral curve}.
\begin{definition}
	A \emph{spectral curve} $\mc{S}$ is a quadruple $\mc{S} = \left(\Sigma,\, x,\, y,\, \omega_{0,2}\right)$ such that:
	\begin{itemize}
		\item $\Sigma$ is a compact Riemann surface;
		\item $dx$ and $dy$ are meromorphic $1$-forms on $\Sigma$;
		\item $\omega_{0,2}$ is a symmetric meromorphic bidifferential on $\Sigma^2$ whose only pole is a double pole on the diagonal with biresidue one.
	\end{itemize}
	The \emph{genus} of a spectral curve is the genus of $\Sigma$. To any spectral curve one associates the not-necessarily meromorphic $1$-form $\omega_{0,1} \coloneq ydx$.
\end{definition}

All spectral curves considered in the present work have genus zero. From here and henceforth, $z$ will denote an affine coordinate on $\Sigma\cong \P^1$ and $(z_0,z_1,\dots,z_{n})$ denotes the corresponding induced coordinates on $\Sigma^{n+1}$. As the multiset $\{z_1,\dots,z_n\}$ appears frequently it will be denoted by $z_{\llbracket n \rrbracket}$. Since the spectral curve is genus zero, there is actually only one choice for the bidifferential $\omega_{0,2}$ \cite{EO07}, namely
\begin{equation}
	\omega_{0,2}(z_1,z_2) = \frac{dz_1dz_2}{(z_1-z_2)^2}\,.
\end{equation}

If $x$ and $y$ are themselves meromorphic, they satisfy identically a polynomial equation 
\begin{equation}
	P(x,y) = 0\,,
\end{equation}
which can be taken to define the spectral curve by denoting the functions $x$ and $y$ as projections onto the respective coordinates. In general, when $x$ and $y$ are not meromorphic, $P$ must be generalised to an entire function in its two variables. 

The topological recursion is not well-defined for an arbitrary spectral curve. Sufficient conditions to define the topological recursion are referred to as \emph{admissibility}.

\begin{definition}
	A spectral curve is called \emph{admissible} if all the zeros of $dx$ are simple, and $dy$ is regular and non-zero at all of these simple zeros.
\end{definition}

\begin{remark}
	The definition of an admissible spectral curve can be vastly generalised and the topological recursion still made sense of. For example, see \cite{BE13,BBCCN18} for non-simple ramification points, \cite{BKS20} for singular curves, or \cite{BKW23} for the case when $x$ and $y$ have essential singularities.
\end{remark}

With this the definition of topological recursion itself may be presented.
\begin{definition}\label{d:TR}
	Given an admissible spectral curve $\mc{S}$ let $R$ denote the zeros of $dx$ and for each $a\in R$ let $\sigma_a(z)$ denote the local Galois involution near $a$, i.e. $x\circ\sigma_a = x$ and $\sigma_a(a) = a$. Denote by $L$ the simple poles of $dy$ and for each $a\in L$ put $\alpha_a \coloneq \Res_{a} dy$. Then define the following infinite tower of $n$-differentials, $\{\omega_{g,n}\}_{g\geq 0,n\geq 1}$, recursively on $2g+n-2$ through the formula
	\begin{multline}\label{e:TR}
		\omega_{g,n+1}(z_0,z_{\llbracket n \rrbracket}) \coloneq \frac{\delta_{n,0}}{2}\Res_{u=0}\frac{du}{u^{2g}} \sum_{a\in L} \Res_{z=a} \left( d_{z} \csch\left( \frac{ u\partial_{x(z)} }{ 2\alpha_a } \right)  \log\left(z-a\right) \right) \int_a^z\omega_{0,2}(z_0,\cdot)\\ 
		+\sum_{a\in R}\Res_{z=a} \left( \frac{ \omega_{g-1,n+2}(z,\sigma_a(z),z_{\llbracket n \rrbracket}) } { (y(\sigma_a(z))-y(z))dx(z) } + \sum_{\substack{Z_1 \sqcup Z_2 = z_{\llbracket n \rrbracket} \\ g_1+g_2 = g \\ (g_i,|Z_i|)\neq (0,0)}} \frac{ \omega_{g_1,|Z_1|+1}(z,Z_1)\omega_{g_2,|Z_2|+1}(\sigma_a(z),Z_2) } { (y(\sigma_a(z))-y(z))dx(z) } \right) \int_a^z\omega_{0,2}(z_0,\cdot) \,,
	\end{multline}
	where $\partial_{x(z)} = (dz/dx(z))\partial_z$, $\sqcup$ denotes the disjoint union, and $u$ is a formal variable; the residue at $u=0$ indicates that the $u^{-1}du$ term in the formal expansion about $u=0$ is picked. In the above formula, the convention $\partial_{x(z)}^{-1}\log(z-a) = 0$ for all $a\in L$ is used.
	
	The $\omega_{g,n}$ are referred to as the \emph{correlators} of the spectral curve $\mc{S}$: if $2g+n-2\leq 0$ they are called \emph{unstable} and if $2g+n-2>0$ they are called \emph{stable}. A point $b$ on $\Sigma$ is called \emph{regular} if all stable correlators have no poles at $b$.
\end{definition}

\begin{remark}\label{r:LogTR}
	The definition of topological recursion given here is what was called logarithmic topological recursion (LogTR) in \cite{ABDKS24}.\footnote{It should be noted that \cite{ABDKS24} uses the convention $\omega_{0,1} = -ydx$, rather than $\omega_{0,1} = ydx$.} The difference from the original Eynard-Orantin formalism was the addition of the first term on the RHS of \eqref{e:TR}, which only directly contributes to the calculation of the $\omega_{g,1}$; the primary motivation for these extra contributions was that the $x$-$y$ swap duality (see \cref{t:x-yswap}) holds more generally in the LogTR setting. In the present work, topological recursion is taken to mean LogTR. 
\end{remark}

\begin{remark}\label{r:reg}
	Observe that points which are poles of both $dx$ and $dy$ are automatically regular. Indeed, for such points the residues at $a\in L$ must vanish.
\end{remark}

The interest in the topological recursion lies, in part, because of the remarkable properties of the $\omega_{g,n}$. Some of these properties are summarised in the following theorem.

\begin{theorem}\label{t:prop}
	The differentials $\omega_{g,n}$ enjoy the following properties:
	\begin{itemize}
		\item the $\omega_{g,n}$ are symmetric in their $n$ variables;
		\item the $\omega_{g,n}$ have vanishing residues at all points;
		\item the $\omega_{g,n}$ have poles only at the zeros of $dx$ and the simple poles of $dy$;
		\item under the rescalings $x\to ax$ and $y\to by$ the correlators transform as $\omega_{g,n} \to (ab)^{2-2g-n} \omega_{g,n}$.
	\end{itemize}
	\begin{proof}
		These properties are well known. See \cite{EO07,ABDKS24}.
	\end{proof}
\end{theorem}

\subsection{Quantising spectral curves}

Quantum curves were motivated in \cref{ss:mot}. Here, to be more precise, let $ \mathcal{S} = \left( \P^1,\, x,\, y,\, \omega_{0,2} \right) $ be an admissible spectral curve. As discussed the functions $x$ and $y$ satisfy a relation $P(x,y)=0$ where $P$ is entire in both its variables. The goal is to `quantise' $P(x,y)$ (replace $x$ and $y$ with non-commuting operators), but before this task may be completed, one first defines the wavefunction this quantisation should kill.

\begin{definition}\label{d:wave}
	Given an admissible spectral curve $\mathcal{S}$ and a regular point $b$ define the \emph{wavefunction}
	\begin{equation}\label{e:wfuncconj}
		\psi(z;b) = \exp\left[ \sum_{n=1}^{\infty} \sum_{g=0}^{\infty} \frac{\hslash^{2g+n-2}}{n!} \int^z_b\cdots\int^z_b \left( \omega_{g,n}  - \delta_{g,0} \delta_{n,2} \frac{dx(z_1) dx(z_2)}{(x(z_1) - x(z_2))^2} \right) \right] \,,
	\end{equation}
	here $ \hbar $ is a formal expansion parameter and there are $n$ integrations in each term. The $(0,1)$ and $(0,2)$ terms may need to be regularised, but this will only introduce a constant factor of the form $e^{\frac{1}{\hbar}A}e^B$ for $A,B\in\C$.
\end{definition}

The statement of the topological recursion/quantum curve correspondence is that there should exist an operator $\hat{P}(\hat{x}, \hat{y}, \hslash )$ such that 
\begin{equation}
	\hat{P}(\hat{x}, \hat{y}, \hslash ) \psi(z;b) = 0\,,
\end{equation}
where $ \hat{x} \coloneq x \cdot $ and $ \hat{y} \coloneq \hbar d/dx$. This $ \hat{P}$ should be a quantisation of the original $P$ in the vague sense that substituting in $\hbar = 0$ and replacing $\hat{x}$ and $\hat{y}$ with their commuting counterparts $x$ and $y$ will bring one back to the original curve $P$. This vague sense is made precise in the following definition. 

\begin{definition}\label{d:Quant}
	Let $\mathcal{S}$ be a spectral curve such that $x$ and $y$ obey $ P(x,y) = 0 $. $\hat{P}(\hat{x},\hat{y};\hslash) $ is called a \emph{quantisation} of $\mc{S}$ (or, equivalently, $P$) if the following expansion exists
	\begin{equation}
		\hat{P}(\hat{x},\hat{y};\hslash) = P(\hat{x},\hat{y}) + \sum_{i=1}^{\infty} \hslash^i\hat{P}_i(\hat{x},\hat{y})\,,
	\end{equation}
	where $P(\hat{x},\hat{y})$ is taken to be normally ordered (in each term all the $\hat{x}$ are put to the left of the $\hat{y}$) and the $ \hat{P}_i $ are  normal ordered polynomials of degree at most $ \deg P - 1 $. The quantisation is referred to as \emph{simple} if all but finitely many of the $P_i$ vanish.
\end{definition}

At last the topological recursion/quantum curve correspondence can be stated precisely.

\begin{conjecture}\label{con:GS}
	Let $\mathcal{S}$ be an admissible spectral curve where $x$ and $y$ satisfy the relation $ P(x,y) = 0 $. Then there exists a regular point $b\in \Sigma$ and a quantisation $\hat{P}(\hat{x},\hat{y}; \hbar)$ of $P(x,y)$ such that
	\begin{equation}
		\hat{P}(\hat{x},\hat{y}; \hslash) \psi(z;b) = 0\,.
	\end{equation}
	$\hat{P}(\hat{x},\hat{y}; \hslash) $ is called the \emph{quantum curve} or \emph{quantum spectral curve}.
\end{conjecture}

\begin{remark}
	It's important to notice that quantum curves are, in general, much simpler than the topological recursion formula (\cref{d:TR}). This constitutes a key motivation for their study \cite{N15}.
\end{remark}

In general, there may be more than one quantisation corresponding to different choices of basepoint $b$ (the integration can also be generalised to an arbitrary degree zero divisor \cite{BE17}). This conjecture has now been proved for many large classes of spectral curves \cite{BE17,EG19,MO20,EGMO21,MO22,BKW23}. However, none of these results allow for $x$ or $y$ to posses $\log$ cuts...

\subsection{$x$-$y$ swap duality}

One can define a topological recursion relation for a $2$-matrix model where one has two fields (matrices) $M_1$ and $M_2$ \cite{CEO06}. In this context there is an obvious duality, namely the interchange of $M_1$ and $M_2$. It is then a natural question to ask how this duality manifests itself in the topological recursion; the answer is that switching $M_1$ and $M_2$ interchanges $x$ and $y$. The so-called $x$-$y$ swap is the abstraction of this duality to general spectral curves and their associated systems of correlators \cite{EO08}. This is made precise in the following definition.

\begin{definition}
	Given a spectral curve $\mathcal{S} = \left(\Sigma,\, x,\, y\, \omega_{0,2}\right)$ define the $x$-$y$ swap dual spectral curve
	\begin{equation}
		\mathcal{S}^{\vee} \coloneq \left(\Sigma^{\vee}\coloneq \Sigma,\, x^{\vee} \coloneq y,\, y^{\vee} \coloneq -x,\, \omega_{0,2}^{\vee}\coloneq \omega_{0,2}\right)\,.
	\end{equation}
	A spectral curve is called \emph{fully admissible} if both $\mc{S}$ and $\mc{S}^{\vee}$ are admissible. 
\end{definition}

Notice that ${\mc{S}^{\vee}}^{\vee} = \mc{S}$ (up to flipping the sign of both $x$ and $y$, which does not affect the topological recursion in any way by \cref{t:prop}); in fact, the two curves $\mathcal{S}$ and $\mathcal{S}^{\vee}$ are intimately connected \cite{EO07,EO08,ABDKS22,ABDKS23,ABDKS24,H23b,H23a,H23c}. It is beyond the scope of the present work to review all such developments, and many of the explicit results in this subfield are very lengthy. In light of this, this subsection will only highlight results that will be useful in the present work; it should not be taken as a general overview.

The fundamental result is the following \cite{ABDKS22,ABDKS24}.
\begin{theorem}\label{t:x-yswap} 
	Let $\mc{S}$ be a fully admissible spectral curve and denote by $\{\omega_{g,n}\}_{g\in\Z_{\geq 0},n\in\Z_{\geq 1}}$ the corresponding system of correlators. Denote by $\{\omega_{g,n}^{\vee}\}_{g\in\Z_{\geq 0},n\in\Z_{\geq 1}}$ the system of correlators associated to $\mathcal{S}^{\vee}$. Then
	\begin{equation}
		\omega_{g,n}^{\vee}(z_{\llbracket n \rrbracket}) = \mathrm{Expr}_{g,n}\left(\left\{\omega_{g',n'}\right\}_{2g'+n'\leq 2g+n},\,\left\{dx(z_i),dy(z_i)\right\}_{i\in \llbracket n \rrbracket},\,\left\{\frac{dx(z_i)dx(z_j)}{\big( x(z_i) - x(z_j) \big)^2}\right\}_{i,j\in \llbracket n \rrbracket}\right),
	\end{equation}
	where $\mathrm{Expr}_{g,n}$ is some explicit expression that, for each $g,n$, is polynomial in its arguments and their derivatives with respect to the $z_i$.
\end{theorem}

The actual form of $\mathrm{Expr}_{g,n}$ involves complicated sums over graphs and will not be needed in the present work. The interested reader may consult \cite{ABDKS22}.

Next, formal Gaussian integration is defined. As will become apparent, this integration is a prominent player in the $x$-$y$ swap duality. In the context of TR this definition previously appeared in \cite[section 6]{ABDKS23}.
\begin{definition}\label{d:FG}
	Given $P\in\C[\xi_1]\cdots[\xi_n][[\hbar^{1/2}]]$ and $Q\in\C[\xi_1]\cdots[\xi_n]$ a non-degenerate quadratic form, define the formal Gaussian integrals
	\begin{equation}
		\int_{\mathrm{FG}}\cdots \int_{\mathrm{FG}} e^{Q(\xi_1,\dots,\xi_n)}P(\xi_1,\dots,\xi_n) d\xi_1\cdots d\xi_n \,,
	\end{equation}
	through the rules
	\begin{equation}
		\begin{split}
			\int_{\mathrm{FG}} e^{-\frac{a}{2}\xi_i^2}\xi_i^{m} d\xi_i \to [m\in 2\Z_{\geq 0}] (m-1)!!\sqrt{\frac{2\pi}{a^{m+1}}}\,,\qquad \int_{\mathrm{FG}} \int_{\mathrm{FG}} \xi_i^{m}\xi_j^{m'}e^{-\xi_i\xi_j} d\xi_i d\xi_j \to \delta_{m,m'} m!\,,
		\end{split}
	\end{equation}
	where for a statement $p$, $[p]$ is the Iverson bracket: it is $1$ if $p$ and $0$ if $\neg p$.
\end{definition}

Now review some of the properties of the preceding definition in the succeeding proposition.

\begin{proposition}\label{p:FGprop}
	As defined above formal Gaussian integration has the following properties:
	\begin{enumerate}
		\item the answer is invariant under a change of variables $\vec{\xi} \to A\vec{\xi} +\mathcal{O}(\hbar^{1/2})$ where $\vec{\xi} = (\xi_1,\dots,\xi_n)$, $A\in \mathrm{GL}_n(\C)$, and the $\mathcal{O}(\hbar^{1/2})$ portion must have polynomial dependence in the $\xi_i$ at each order in $\hbar^{1/2}$ (see \cite[section 6]{ABDKS23});
		\item when performing formal Gaussian integration over multiple variables the order of integration does not matter;
		\item integration by parts holds for each $\xi$
		\begin{equation}
			\int_{\mathrm{FG}} \left(\frac{d}{d\xi}f(\xi;\hbar)\right)g(\xi;\hbar) d\xi = -\int_{\mathrm{FG}} f(\xi;\hbar) \left(\frac{d}{d\xi}g(\xi;\hbar)\right) d\xi \,,
		\end{equation}
		where it is assumed $f$ and $g$ multiply to give the form of the integrand required by \cref{d:FG};
		\item if for two functions $f$ and $f^{\vee}$
		\begin{equation}
			f(w;\hbar) = \frac{i}{\sqrt{2\pi\hbar}} \int_{\mathrm{FG}} \exp \left( \frac{1}{\hbar}\left(x(w)(y(w)-y(z)) + \int_w^z xdy\right) \right) f^{\vee}(z;\hbar) y'(z) dz \,,
		\end{equation}
		then it follows that (this is \cite[proposition 3.9]{ABDKS23})
		\begin{equation}
			f^{\vee}(z;\hbar) = \frac{1}{\sqrt{2\pi\hbar}} \int_{\mathrm{FG}} \exp \left( -\frac{1}{\hbar}\left(y(z)(x(z)-x(w)) + \int_z^w ydx\right) \right) f(w;\hbar) x'(w) dw \,;
		\end{equation}
		\item if $P_N\in \C[\xi_1]\cdots[\xi_n][[\hbar^{1/2}]]$ and $Q_N\in\C[\xi_1]\cdots[\xi_n]$ are convergent sequences\footnote{Convergence in the space of $\hbar^{1/2}$ formal series means convergence in each order of $\hbar^{1/2}$ (i.e., the weak topology).} then
		\begin{equation}
			\lim\limits_{N\to\infty} \int_{\mathrm{FG}}\cdots \int_{\mathrm{FG}} e^{Q_N(\xi_1,\dots,\xi_n)}P_N(\xi_1,\dots,\xi_n) d\xi_1\cdots d\xi_n = \int_{\mathrm{FG}}\cdots \int_{\mathrm{FG}} e^{Q(\xi_1,\dots,\xi_n)}P(\xi_1,\dots,\xi_n) d\xi_1\cdots d\xi_n \,,
		\end{equation}
		where $P_N\to P$ and $Q_N\to Q$.
	\end{enumerate}
\end{proposition}
\begin{proof}
	The first statements is known and appears in \cite{ABDKS23}, while the second statement follows from the first by diagonalising $Q$; the third statement also follows from the first (again diagonalise $Q$) and noting that, by direct computation 
	\begin{equation}
		\int_{\mathrm{FG}} d\left(e^{-\frac{a}{2}\xi^2}\xi^m\right) = 0\,.
	\end{equation}
	The fourth statement is \cite[proposition 3.9]{ABDKS23}. The fifth statement may be proven by induction $n$. For $n=1$, the statement follows essentially immediately from the fact that each power of $\hbar^{1/2}$ has only finitely many terms in $\xi_1$ and the first rule given in \cref{d:FG}. For the induction step, one may apply the first property to change variables so $Q_N(\xi_1,\dots,\xi_n)+\frac{1}{2}a_N\xi_n^2\in \C[\xi_1]\cdots[\xi_{n-1}]$ for a convergent sequence $a_N\in C$. The proof is then completed using the same logic on the integration over $\xi_n$ as was used in the induction beginning. 
\end{proof}

To illustrate how the formal Gaussian integration is used in practice, the following example is provided (see also \cite[equation (6.6)]{ABDKS23}).

\begin{example}\label{ex:Int}
	Examine the following formal Gaussian integral, where $f(w;\hbar)$ is a formal series in $\hbar$ with coefficients that are meromorphic functions in $w$
	\begin{equation}
		f^{\vee}(z;\hbar) \coloneq \frac{1}{\sqrt{2\pi\hbar}} \int_{\mathrm{FG}} \exp\left(-\frac{1}{\hbar}\left(y(z)(x(z)-x(w)) + \int_z^w ydx\right)\right)f(w;\hbar)x'(w) dw \,.
	\end{equation}
	Notice the saddle point occurs at $w = z$, so define the integration variable $\xi = (w-z)/\hbar^{1/2}$ (the $\hbar^{-1/2}$ factor is chosen so the quadratic, $\xi^2$, term in the exponential has no $\hbar$ dependence). Expanding around $\xi = 0$
	\begin{equation}
		-\frac{1}{\hbar}\left(y(z)(x(z)-x(w)) + \int_z^w ydx\right) = -\frac{x'(z)y'(z)}{2}\xi^2 -  g(z,\xi;\hbar)\,,
	\end{equation}
	where $g(z,\xi;\hbar) = \mathcal{O}(\xi^3)$ is given by
	\begin{equation}
		g(z,\xi;\hbar) = \sum_{m=2}^{\infty}\left(\frac{d^m}{dz^m}\big[y(z)x'(z)\big] - y(z)\frac{d^{m+1}}{dz^{m+1}}x(z)\right)\frac{\xi^{m+1}\hbar^{\frac{m-1}{2}}}{(m+1)!}\,.
	\end{equation}
	Then one can define coefficients $c_m$ through the following expansion
	\begin{equation}
		e^{-g(z,\xi;\hbar)}f(z+\hbar^{1/2}\xi;\hbar)x'(z+\hbar^{1/2}) = \sum_{m=0}^{\infty} c_m(z;\hbar) \xi^m \,,
	\end{equation}
	where, for fixed $z$, it follows that $c_m(z;\hbar)\in\C[[\hbar^{\frac{1}{2}}]]$. Finally, one can expand around $\xi=0$ and integrate term-by-term
	\begin{equation}
		f^{\vee}(z;\hbar) = \sum_{m=0}^{\infty} \frac{c_m(z;\hbar)}{\sqrt{2\pi}} \int_{\mathrm{FG}} e^{-\frac{1}{2}\frac{x'(z)y'(z)}{2}\xi^2} \xi^m d\xi \,.
	\end{equation}
	It is important to notice the following facts:
	\begin{itemize}
		\item $c_m(z;\hbar) = \mathcal{O}(\hbar^{m/2})$ as each derivative in $\xi$ comes with a factor of $\hbar^{1/2}$;
		\item from the above fact it follows that the coefficient of each power of $\hbar^{1/2}$ is polynomial in $\xi$ as required by \cref{d:FG} (in particular there is no issue with commuting the integral and sum over $m$);
		\item if $f(z;\hbar)$ only involves integer powers of $\hbar$ one concludes $f^{\vee}(z;\hbar)$ will also only have integer powers of $\hbar$ as all the odd $m$ in the above expansion integrate to zero, after observing that the $c_m$ with even $m$ can be viewed as coming from even number of derivatives with respect to $\xi$. Indeed, in this case $c_m(z;\hbar) \in \hbar^{m/2}\C[[\hbar]]$.
	\end{itemize}
	
	One can also consider other coordinate choices, for example $\xi = \log(w/z)/\hbar^{1/2}$. Here, the expansion around $\xi=0$ in the argument of the exponential reads
	\begin{equation}
		-\frac{1}{\hbar}\left(y(z)(x(z)-x(w)) + \int_z^w ydx\right) = -\frac{z^2x'(z)y'(z)}{2}\xi^2 -  \tilde{g}(z,\xi;\hbar)\,,
	\end{equation}
	where $\tilde{g} = \mathcal{O}(\hbar^{1/2}\xi^3)$ is a formal power series in $\hbar^{1/2}$. One can then define coefficients $\tilde{c}_m$ through
	\begin{equation}
		e^{-\tilde{g}(z,\xi;\hbar)}f(ze^{\hbar^{1/2}\xi};\hbar)x'(ze^{\hbar^{1/2}\xi})e^{\hbar^{1/2}\xi} = \sum_{m=0}^{\infty} \tilde{c}_m(z;\hbar) \xi^m \,,
	\end{equation}
	and integrate term-by-term in $m$ as before. 
\end{example}

Next, to a given spectral curve a certain kernel is associated. In the present work, the reason this kernel is of interest is because it is both intimately related to the wavefunction $\psi$ and known to transform nicely under the $x$-$y$ swap.

\begin{definition}
	Given a system of correlators $\omega_{g,n}$ define the \emph{Baker-Akhiezer kernel}\footnote{Strictly speaking, this should only be called the Baker-Akhiezer kernel when the $\omega_{g,n}$ posses the so-called KP-integrability property. Most genus zero curves posses this property. See \cite{ABDKS23,ABDKS24b} for details.}
	\begin{multline}
		K(z_1,z_2) = \frac{1}{x(z_1)-x(z_2)}\exp\left(\frac{1}{2!}\int_{z_2}^{z_1}\int_{z_2}^{z_1}\left(\omega_{0,2}(w_1,w_2)-\frac{dx(w_1)dx(w_2)}{\big(x(w_1)-x(w_2)\big)^2}\right)\right. \\
		\left. + \sum_{2g+n-2>0}\frac{\hbar^{2g+n-2}}{n!}\int_{z_2}^{z_1}\cdots\int_{z_2}^{z_1} \omega_{g,n}\right)\,.
	\end{multline}
	Define the \emph{dual Baker-Akhiezer kernel}, $K^{\vee}(z_1,z_2)$, in the obvious way by making the replacements $x\to x^{\vee}$ and $\omega_{g,n}\to \omega_{g,n}^{\vee}$ in the above formula.
\end{definition}

The following result is \cite[theorem 3.12]{ABDKS23}:\footnote{The sign conventions in \cite{ABDKS22,ABDKS23} differ from what is presented here; in particular the sign of $\omega_{0,1}$ is flipped whereas the sign of $\omega_{0,1}^{\vee}$ uses the same convention. By the last property listed in \cref{t:prop}, this swaps the arguments in the Baker-Akhiezer kernel.}

\begin{theorem}\label{t:BAkernel}
	Given a fully admissible spectral curve $\mc{S}$ it follows that
	\begin{multline}
		K^{\vee}(z_1,z_2) = \frac{-i}{2\pi\hbar}\int_{\mathrm{FG}}\int_{\mathrm{FG}} K(w_1,w_2) x'(w_1)x'(w_2) \exp\left(\frac{1}{\hbar}\left(y(z_1)(x(z_1)-x(w_1))+\int_{z_1}^{w_1}ydx\right)\right) \\ \times \exp\left(-\frac{1}{\hbar}\left(y(z_2)(x(z_2)-x(w_2))+\int_{z_2}^{w_2}ydx\right)\right) dw_1dw_2 \,.
	\end{multline}
\end{theorem}
\begin{proof}
	This follows from the fact that the two systems of correlators $\omega_{g,n}$ and $\omega_{g,n}^{\vee}$ are related via \cref{t:x-yswap}. In fact, any two systems of correlators related by the same formula \cref{t:x-yswap} and having no poles on the diagonal (except for the double pole present in $\omega_{0,2}=\omega_{0,2}^{\vee}$) will satisfy this theorem. As the definition of the topological recursion used here, \cref{d:TR}, respects the $x$-$y$ swap relation \cite{ABDKS24}, the theorem holds. The interested reader should consult \cite[theorem 4.3]{ABDKS23} for the detailed proof.
\end{proof}

\begin{remark}
	In the preceding theorem the saddle points are at $w_i = z_i$. Various coordinate choices are possible, but it is assumed that the coordinate choice made leaves no explicit $\hbar$ dependence in the resulting quadratic form $Q(\xi_1,\xi_2)$ in the exponent, as required by \cref{d:FG}; this requirement ensures that each term in the corresponding expansion in $\hbar^{1/2}$ has only polynomial dependence in each of the $\xi_i$. For example, one such coordinate choice is $\xi_i = (w_i-z_i)/\hbar^{1/2}$. See \cref{ex:Int} for an instructive demonstration of the sort of formal Gaussian integration used in the above theorem.
\end{remark}

\section{The Laplace transform}\label{s:Lap}

Begin by defining what is meant by Laplace transform.

\begin{definition}\label{d:Lap}
	Let $\mc{S}$ be a spectral curve and $U,V\subset \Sigma$ be open sets such that $x$ and $y$ are holomorphic on $U$ and $V$. The Laplace transform is defined as a map
	\begin{equation}
		\begin{split}
		&\mathscr{L}_{\mc{S}} \, :\, e^{\frac{1}{\hbar}\C} e^{\frac{1}{\hbar}\int ydx}\mc{M}(V)[[\hbar]] \to e^{\frac{1}{\hbar}\C} e^{-\frac{1}{\hbar}\int xdy}\mc{M}(U)[[\hbar]] \,,\\
		&\mathscr{L}_{\mc{S}}\left[e^{\frac{1}{\hbar}C} e^{\frac{1}{\hbar}\int^w ydx} f(w;\hbar)\right](z) \coloneq \frac{e^{\frac{1}{\hbar}C}}{\sqrt{2\pi\hbar}} \int_{\mathrm{FG}} e^{\frac{1}{\hbar}\int^w ydx} e^{-\frac{x(w)y(z)}{\hbar}} f(w;\hbar) x'(w) dw \,,
		\end{split}
	\end{equation}
	where $w\in V$ and $z\in U$.
\end{definition}

The key properties of this Laplace transform in the context of the topological recursion and the $x$-$y$ swap duality are given by the following lemmata and theorem, which is the main theorem of the present work.

\begin{lemma}\label{l:LapInv}
	The Laplace transform as defined above has the inverse
	\begin{equation}
		\mathscr{L}^{-1}_{\mc{S}}\left[e^{\frac{1}{\hbar}C} e^{-\frac{1}{\hbar}\int^z xdy} f^{\vee}(z;\hbar)\right](w) = \frac{ie^{\frac{1}{\hbar}C}}{\sqrt{2\pi\hbar}} \int_{\mathrm{FG}} e^{-\frac{1}{\hbar}\int^z xdy} f^{\vee}(z;\hbar) y'(z) dz\,.
	\end{equation}	
\end{lemma}
\begin{proof}
	This is an immediate corollary of \cref{p:FGprop} (\cite[proposition 3.9]{ABDKS23}). 
\end{proof}

\begin{lemma}\label{l:LapParts}
	Let $\Phi(w;\hbar)$ be in the domain of $\ms{L}_{\mc{S}}$ and let $\Phi^{\vee}(z;\hbar)$ be its Laplace transform.
	\begin{equation}
		\begin{split}
			\ms{L}_{\mc{S}}\left[\hat{y}(w)^k \hat{x}(w)^l \Phi(w;\hbar)\right](z) &= \hat{x}^{\vee}(z)^k \hat{y}^{\vee}(z)^l \Phi^{\vee}(z;\hbar)\,, \\
			\ms{L}_{\mc{S}}^{-1}\left[\hat{y}^{\vee}(z)^k \hat{x}^{\vee}(z)^l \Phi^{\vee}(z;\hbar)\right](w) &=  \hat{x}(w)^k \hat{y}(w)^l \Phi(w;\hbar)\,,
		\end{split}
	\end{equation}
	where $\hat{x}(w) \coloneq x(w)\cdot$, $\hat{y}(w) \coloneq \hbar d/dx(w)$, $\hat{x}^{\vee}(z) \coloneq y(z)\cdot$, and $\hat{y}^{\vee}(z) \coloneq -\hbar d/dy(z)$.
\end{lemma}
\begin{proof}
	The $l=0$ identities follow immediately via integration by parts (\cref{p:FGprop}) and the definition of the Laplace transform. One can then derive the $l\neq 0$ identities using the fact that $\ms{L}_{\mc{S}} \circ \ms{L}_{\mc{S}}^{-1} = \ms{L}_{\mc{S}}^{-1} \circ \ms{L}_{\mc{S}}$ is the identity, by the previous lemma.
\end{proof}

\begin{theorem}\label{t:Lap}
	Let $\mc{S}$ be a fully admissible spectral curve. Assume there exists a point $b\in \Sigma$ such that both $dx(z)$ and $dy(z)$ have a pole at $b$. Define $\psi^{\vee}(z;b)$ in the obvious way by making the replacements $\omega_{g,n}\to\omega_{g,n}^{\vee}$ and $x \to x^{\vee}$ in \cref{d:wave}. Then
	\begin{equation}
			\psi^{\vee}(z;b) =  \mathscr{L}_{\mc{S}}[\psi(w;b)](z)\,,\qquad \psi(w;b) = \mathscr{L}_{\mc{S}}^{-1}[\psi^{\vee}(z;b)](w)\,,
	\end{equation}
	where the wavefunctions $\psi$ and $\psi^{\vee}$ may need to be multiplied by a constant from the form given in \cref{d:wave} so no constant of proportionality is introduced in the above formula.
\end{theorem}
\begin{proof}
	The essential idea of the proof is as follows. One starts with the statement of \cref{t:BAkernel}; a new integration coordinate $\xi = \xi(w_2,z_2)$ where $\xi(w_2=z_2,z_2) = 0$ is the saddle point in the integration over $w_2$ is then defined. Next, one expands around $\xi = 0$ as explained in \cref{ex:Int}. The $\xi = 0$ term (this term amounts to evaluating the integrand at $w_2=z_2$) will, up to a constant of proportionality, give the Laplace transform of the wavefunction (or its dual, depending on whether one starts from the Laplace transform or its inverse) in the limit $z_2\to b$. Finally, one shows the rest of the terms in the expansion about $\xi=0$ vanish in the limit $z_2 \to b$, which yields the theorem. This program is executed in two cases, where two different choices of the coordinate $\xi$ are required (these two choices of $\xi$ are the ones illustrated in \cref{ex:Int}).
	
	\medskip
	\noindent
	\emph{Case I: one or both of $dx$ and $dy$ have a non-simple pole at $b$.}\par

	If $b=\infty$ in the affine coordinate $z$ redefine the coordinate so $b\neq\infty$. Observe that $b$ must be a regular point (see \cref{r:reg}). Notice that, by definition
	\begin{equation}
		K(z_1,z_2) = \frac{e^{-\frac{1}{\hbar}\int^{z_1}_{z_2} \omega_{0,1}}\psi(z_1,z_2)}{x(z_1)-x(z_2)} \,.
	\end{equation}
	Define $\phi(z_2) \coloneq -x(z_2) e^{-\frac{1}{\hbar}\int^{z_2}_* \omega_{0,1}}$ where $*$ is chosen so that
	\begin{equation}
		\psi(z;b) = \lim\limits_{z_2 \to b} e^{\frac{1}{\hbar}\int^{z}_{z_2} \omega_{0,1}} \frac{K(z,z_2)(x(z)-x(z_2))}{\phi(z_2)} \,.
	\end{equation}
	Similarly, define $\phi^{\vee}$ in an analogous manner. Now, from the statement of \cref{t:BAkernel}, and defining a new integration variable through $\xi = (w_2-z_2)/\sqrt{h}$
	\begin{equation}\label{e:CaseIComp}
		\begin{split}
			\psi^{\vee}(z;b) = & \lim\limits_{z_2\to b} \frac{-i}{2\pi\hbar} \int_{\mathrm{FG}} \int_{\mathrm{FG}} K(w,w_2) x'(w) x'(w_2) \frac{y(z_2)-y(z)}{\phi^{\vee}(z_2)} \exp\left(\frac{1}{\hbar}\int_{w_2}^{w}ydx\right) \\ & \qquad \qquad \times \exp\left(\frac{1}{\hbar}\big(y(z_2)x(w_2)-y(z)x(w)\big)\right) dw dw_2 \\
			= & \lim\limits_{z_2\to b} A(z_2)e^{\frac{1}{\hbar}B(z_2)} \frac{1}{\sqrt{2\pi\hbar}} \int e^{-\frac{y(z)x(w)}{\hbar}} \psi(w;z_2)  x'(w) \\ 
			& \qquad \qquad \times \sqrt{\frac{x'(z_2)y'(z_2)}{2\pi}} \sum_{m=0}^{\infty} \int_{\mathrm{FG}} e^{-\frac{x'(z_2)y'(z_2)}{2}\xi^2} c_m(w,z_2;\hbar) \xi^m d\xi dw \,,
		\end{split}
	\end{equation}
	where $c_0 = 1$\footnote{The reader should beware that the $c_m$ in \cref{ex:Int} were not normalised this way.} and $A,B$ have no $\hbar$ dependence. In the above computation the fact that $b$ is a pole of $dx$ and $dy$ has been used to establish $y(z_2)-y(z) \sim y(z_2)$ and $(x(w_2)-x(w))|_{w_2=z_2}\sim x(z_2)$ as $z_2\to b$.
	
	Now, each $c_m$ is computed by taking one further $\xi$ derivative versus the computation for $c_{m-1}$ and evaluating at $\xi = 0$ (equivalently, $w_2=z_2$). Each such derivative will increase the order of the pole in the integrand at $z_2=b$ by at most one. However, each increasing power of $\xi$, after integration, gives a factor of $(x'(z_2)y'(z_2))^{-1/2}$, which goes to zero, and as either $dx$ or $dy$ has a non-simple pole at $b$ each succeeding term in the series over $m$ must go to zero faster than the preceding one.
	
	Now recall that each $c_m$ is a power series in $\hbar^{1/2}$ with leading order $\hbar^{m/2}$. Thus, at each order in $\hbar^{1/2}$, only finitely many terms in the sum over $m$ contribute. Therefore, there is no issue with commuting the limit as $z_2\to b$ with the sum over $m$.
	
	Next, as the LHS is by definition finite the RHS must also be finite. The $A(z_2)$ factor may therefore by absorbed into the regularisation of the $(g,n)=(0,2)$ term in $\psi$ as $z_2\to b$ and similarly for $B(z_2)$ with the regularisation of the $(g,n) = (0,1)$ term, to give an overall answer that is finite as $z_2\to b$. This gives the formula for $\psi^{\vee}$ in terms of the Laplace transform of $\psi$. The second equality immediately follows by \cref{p:FGprop}.
	
	\medskip
	\noindent
	\emph{Case II: both of $dx$ and $dy$ have a simple pole at $b$.}\par
	
	In this case, change the affine coordinate so that $b=0$ and $z_2=\infty$ is a regular point which is neither a zero or a pole of $dx$ or $dy$. Define the integration variable $\xi = \hbar^{-1/2}\log(w_2/z_2)$ and observe (the computation is entirely analogous to what was done in \cref{e:CaseIComp} where an additional step was shown)
	\begin{equation}\label{e:CaseIIComp}
		\begin{split}
			\psi^{\vee}(z;b) & = \lim\limits_{z_2\to b} \tilde{A}(z_2)e^{\frac{1}{\hbar}\tilde{B}(z_2)} \frac{1}{\sqrt{2\pi\hbar}} \int e^{-\frac{y(z)x(w)}{\hbar}} \psi(w,z_2)  x'(w)  \int_{\mathrm{FG}} e^{-\frac{z_2^2x'(z_2)y'(z_2)}{2}\xi^2} \Xi(\xi,z_2,w;\hbar) d\xi dw \,,
		\end{split}
	\end{equation}
	where $\Xi(\xi,z_2,w;\hbar)$ is a formal series in $\hbar^{1/2}$. Notice that
	\begin{equation}
		\frac{d}{d\xi} = w_2\frac{d}{dw_2}\,,
	\end{equation}
	so, for a function $f(w_2;\hbar)$, if $f$ has neither a zero nor a pole at $w_2=0$ then $d^mf/d\xi^m|_{w_2=z_2}$ will have a zero at $z_2 = 0$ for all $m\geq 1$. Thus, if it is established that $\Xi$ has neither a zero nor a pole at $w_2=0$, then the leading order term in the expansion about $\xi=0$ will be the only one that contributes in the limit $z_2\to b=0$ by similar logic as in case I\footnote{The perspicacious reader will notice this is not quite correct. Expressions whose poles in the $w_2$-plane depend on $z_2$, such as $1/(w_2/z_2+1)$, but have no pole or zero at $z_2=0$ after substituting $w_2\to z_2$, may not behave under differentiation in the claimed manner. Fortunately, the five different contributing factors do not present these issues. For all factors save the integrals of the stable correlators, this will be clear through explicit computation, and for the stable correlators one notes that the properties listed in \cref{t:prop} will imply their integrals are meromorphic functions such that the location of their poles in the $w_2$ plane do not depend on the value of $z_2$.} (observe that the coefficient of $-\xi^2/2$ in the argument of the exponential in the integral over $\xi$ in \eqref{e:CaseIIComp}, $z_2^2x'(z_2)y'(z_2)$, is finite and non-zero as $z_2\to 0$).
	
	Schematically, the factors contributing to $\Xi$ are:
	\begin{enumerate}
		\item $dx(w_2)/d\xi$;
		\item $\exp\left(\text{integrals of stable correlators}\right)$;
		\item $\exp\left(\frac{1}{\hbar}[y(z_2)x(w_2) - \int^{w_2} ydx + \frac{1}{2}z_2^2y'(z_2)x'(z_2)\xi^2 -\tilde{B}(z_2)]\right)$;
		\item $\exp\left(\frac{1}{2} \int \int [\omega_{0,2} - x_1^*x_2^*\omega_{0,2}]\right)$, where $x_i^*$ is the pullback in the $i^\text{th}$ variable;
		\item $1/(x(w_1)-x(w_2))$.
	\end{enumerate}
	
	The first factor is finite and non-zero at $w_2=0$ as $dx$ has a simple pole at $w_2=0$ by assumption. The argument of the exponential in the second factor is finite at $w_2=0$ as $b$ is a regular point; the exponential of a finite number is a finite and non-zero number. For the third factor, observe that the second derivative of the exponent with respect to $\xi$ is
	\begin{equation}
		-w_2^2y(w_2)x''(w_2)+w_2^2y(z_2)x''(w_2)-w_2^2 x'(w_2)y'(w_2) +z_2^2x'(z_2)y'(z_2) - w_2y(w_2)x'(w_2)+w_2y(z_2)x'(w_2) \,.
	\end{equation}
	Noting that, as $w_2\to 0$, $x(w_2) \sim \alpha_0^{\vee}\log(w_2) + \mathcal{O}(1)$ and $y(w_2)\sim \alpha_0 \log(w_2)+ \mathcal{O}(1)$ for some $\alpha_0^{\vee},\alpha_0\in\C^*$, one sees that the apparent logarithmic divergence in the above expression as $w_2\to 0$ actually cancels. Thus, the exponential of this expression has neither a zero nor a pole at $w_2=0$, as hoped for.
	
	Turn now to the contribution from $\omega_{0,2}$. Write out this factor explicitly
	\begin{multline}
		\exp\left(\frac{1}{2}\left(\int_{w_2}^{w}\int_{w_2}^{w}-\int_{z_2}^{w}\int_{z_2}^{w}\right)\frac{dt_1dt_2}{(t_1-t_2)^2} - \frac{dx(t_2)dx(t_2)}{(x(t_1)-x(t_2))^2}\right) = \sqrt{\frac{x'(z_2)}{x'(w_2)}}\frac{(x(w)-x(w_2))(w-z_2)}{(w-w_2)(x(w)-x(z_2))}\,.
	\end{multline}
	The $x(w)-x(w_2)$ in the numerator precisely cancels the fifth factor in the list. The $(w-w_2)^{-1}$ factor is of the desired form. However, the $x'(w_2)^{-1/2}$ is clearly not of the desired form. To remedy this situation one can perform the trivial rewriting
	\begin{equation}
		\sqrt{\frac{1}{x'(w_2)}} = \sqrt{\frac{e^{\hbar^{1/2}\xi}}{x'(w_2)}}e^{-\frac{1}{2}\hbar^{1/2}\xi}\,.
	\end{equation}
	Then, by the definition of the coordinate $\xi$ the $\sqrt{e^{\hbar^{1/2}\xi}/x'(w_2)}$ factor is of the desired form. The $e^{-\frac{1}{2}\hbar^{1/2}\xi}$ can be accounted for by shifting the coordinate $\xi$ by $\hbar^{1/2}/[2z_2^2x'(z_2)y'(z_2)]$ so it is absorbed into the quadratic term in the exponential (i.e., complete the square). This introduces a factor $a(\hbar)$ in the relation between $w_2$ and $\xi$ so $w_2 = a(\hbar^{1/2})z_2e^{\hbar^{1/2}\xi}$ with $a(0)=1$. However, $z_2\to 0$ still means $w_2\to 0$ in these new coordinates so this won't affect any of the previous arguments. This also multiplies the resulting wavefunction by a constant, but this constant is finite as $z_2^2x'(z_2)y'(z_2)$ is non-zero and finite in the limit $z_2 \to b$.
	
	Thus, the leading order coefficient in the expansion of $\Xi$ about $\xi=0$ is indeed of higher order as $z_2\to 0$ than the rest of the coefficients. From this point on, the proof of case II proceeds, \textit{mutatis mutandis}, as in case I.
\end{proof}

\begin{remark}
	In \cite{H23b} it was noted that the wavefunction $\psi$ and the dual wavefunction $\psi^{\vee}$ should, intuitively, be Laplace transforms of one another, and multiple affirmative examples were provided to support this intuition. The above theorem translates this into a rigorous result.
\end{remark}

\begin{remark}\label{r:Highg}
	For higher genus curves, the above proof should hold for what is called the `perturbative wavefunction'. However, in these cases, the actual wavefunction requires nonpeturbative corrections \cite{EGMO21}. It would be interesting to examine in more detail how the above theorem works in higher genus.
\end{remark}

\section{Quantum curves}\label{s:qc}

\subsection{General consequences of the Laplace transform}\label{ss:gen}

The most obvious application of the Laplace transform is in the derivation of quantum curves. Indeed, assume that one has a fully admissible spectral curve $\mc{S}$ and is in possession of a dual quantum curve (i.e., a quantum curve for $\mc{S}^{\vee}$) $\hat{P}^{\vee}$ so that 
\begin{equation}\label{e:QCex}
	\hat{P}^{\vee}(\hat{x}^{\vee},\hat{y}^{\vee};\hbar)\psi^{\vee}(z;b) = 0\,, 
\end{equation}
where $\hat{y}^{\vee} \coloneq -\hbar d/dy$ and $\hat{x}^{\vee} \coloneq y\cdot$. Furthermore, assume that $b$ happens to be a pole of $dx$ and $dy$. Then, using \cref{l:LapParts}, one can immediately get a quantum curve for $\mc{S}$ by taking the inverse Laplace transform of \cref{e:QCex}. Of course, in complete generality, finding the quantum curve for $\mc{S}^{\vee}$ is no more or less challenging than finding the quantum curve for $\mc{S}$. But if $y$ happens to be unramified, the recursive structure of topological recursion collapses, and one can compute the dual wavefunction $\psi^{\vee}$ explicitly. Then, a differential equation for $\psi^{\vee}$ can be constructed and the Laplace transform taken. This approach is illustrated in the following proposition.

\begin{proposition}
	Let $\mc{S} = \left(\P^1,\, x(z) = P_1(z)/P_2(z),\, y(z) = z ,\,\omega_{0,2}(z_1,z_2) = dz_1dz_2/(z_1-z_2)^2 \right)$ be a fully admissible spectral curve where $P_1$ and $P_2$ are coprime polynomials such that $\deg P_1>\deg P_2$. It follows that
	\begin{equation}
		\big[P_1(\hat{y})-P_2(\hat{y})\hat{x}\big]\psi(z;\infty) = 0\,.
	\end{equation} 
\end{proposition}
\begin{proof}
	Notice that $\psi^{\vee}(z;\infty)$ satisfies
	\begin{equation}
		\left[P_1(\hat{x}^{\vee})-P_2(\hat{x}^{\vee})\hat{y}^{\vee}\right]\psi^{\vee}(z;\infty)=0\,,
	\end{equation}
	where $\hat{y}^{\vee} \coloneq -\hbar d/dy$ and $\hat{x}^{\vee} \coloneq y\cdot$. Taking the Laplace transform yields the desired result.
\end{proof}

Although this is a very slick, and new, derivation of the quantum curve, systematic ways to derive quantum curves with $y=z$ are known for any meromorphic $x$ \cite{BE17} and even in many cases when $x$ is not meromorphic \cite{BKW23}. It should be noted, however, that for appropriate basepoints $b$, this new method is easier. What is more interesting is the case $y=\log z$.\footnote{Up to changing coordinates, these are the only two cases of injective, and therefore unramified, $y$ where $dy$ is meromorphic.}

\begin{remark}
	In a sense the results here provide a very intuitive explanation for why quantum curves with simple $y(z)$ have easy systematic processes to derive quantum curves with few $\hbar$ corrections \cite{BE17}, whereas quantum curves with more complicated $y$ require more sophisticated methods to derive quantum curves with many $\hbar$ corrections \cite{EGMO21}.
\end{remark}

Of particular interest amongst the curves with $y=\log z$ are those of the form $P_2(e^y)e^x-P_1(e^y) = 0$ where $P_1$, $P_2$ are coprime polynomials. These curves have relations to mirror symmetry \cite{BKMP08,BM08,GS11,BS12,Z12,DOSS14,EO15,FLZ17,FLZ20} and knot theory \cite{ABM12,BE12,GJKS14,DPSS19}. In \cite{GS11} it was conjectured that the dual quantum curve takes the form
\begin{equation}
	\left[P_2(e^{\hat{x}^{\vee}-\hbar/2})e^{\hat{y}^{\vee}+\hbar/2}-P_1(e^{\hat{x}^{\vee}-\hbar/2})\right] \psi^{\vee}(z;b) = 0 \,.
\end{equation}

It was argued by Hock \cite{H23b} that such a quantum curve should kill the dual wavefunction as a motivation for defining a modified version of the topological recursion, which would become LogTR in \cite{ABDKS24} (see \cref{r:LogTR}). Hock's original proposal indeed aligns with LogTR in many cases \cite[corollary 3.8]{ABDKS24}. Taking the LogTR as given, one can repeat a similar argument to that of Hock, and then take the Laplace transform to obtain a quantum curve for the original spectral curve. This is the content of the following theorem.

\begin{theorem}\label{t:QC}
	Given an admissible spectral curve $\mc{S}$ where $x$ and $y$ satisfy a relation of the form $P_2(e^y)e^x-P_1(e^y) = 0$ such that $P_1$ and $P_2$ are coprime polynomials satisfying:
	\begin{enumerate}
		\item $\deg P_1 \neq \deg P_2$;
		\item $\deg \left[z^{\deg P_1}P_1(z^{-1})\right] \neq \deg \left[z^{\deg P_2}P_2(z^{-1})\right]$;
		\item $P_1/P_2$ has only simple zeros and poles on $\C^*$,
	\end{enumerate}
	it follows that there exists a quantum curve for the basepoint $e^{y(b)} = 0,\infty$
	\begin{equation}
		\left[P_2(e^{\hat{y}-\hbar/2})e^{\hat{x}+\hbar/2}-P_1(e^{\hat{y}-\hbar/2})\right] \psi(z;b) = 0 \,.
	\end{equation}
\end{theorem}
\begin{proof}
	Choose the uniformisation coordinate $z$ so that $y(z) = \log(z)$. By the restrictions on the degrees of $P_1$ and $P_2$ both $dx$ and $dy$ have simple poles at $b=0,\infty$. One now wishes to compute the dual wavefunction $\psi^{\vee}(z;b)$. Based on suggestive computations made in \cite{H23b}, \cite[corollary 3.8]{ABDKS24} proved the following
	\begin{equation}
		\sum_{g=0}\hbar^{2g-1}\omega^{\vee}_{g,1}(z_1) = -\frac{1}{2}d_z \csch(\frac{1}{2}\hbar\partial_{y(z)}) x(z)\,,
	\end{equation}
	where $\partial_{y(z)}^{-1}$ is naturally interpreted as integrating with respect to $dy(z)$ (the integration constant will play no role here). Thus the dual wavefunction is
	\begin{equation}
		\psi^{\vee}(z;b) = \exp\left(\frac{1}{2}y(z) - \frac{1}{2}\csch(\frac{1}{2}\hbar\partial_{y(z)}) x(z)\right) \,.
	\end{equation}
	Now observe that $e^{a\partial_{y(z)}}f(z) = f(e^a z)$ for any constant $a$. Next, given the fact that $P_2(e^{y(z)-\hbar/2})/P_1(e^{y(z)-\hbar/2}) = e^{x(e^{-\hbar/2}z)}$, one computes
	\begin{equation}
		\frac{P_2(e^{y(z)-\hbar/2})}{P_1(e^{y(z)-\hbar/2})}e^{-\hbar\partial_{y(z)}+\hbar/2}\psi^{\vee}(z;b) = e^{\frac{1}{2}y(z) - x(e^{-\hbar/2}z) - \frac{1}{2}\csch(\frac{1}{2}\hbar\partial_{y(z)}) x(e^{-\hbar}z)} = e^{\frac{1}{2}y(z) - \frac{1}{2}\csch(\frac{1}{2}\hbar\partial_{y(z)}) x(z)}\,,
	\end{equation}
	where the last equality follows from
	\begin{equation}
		\frac{e^{-u}}{e^{\frac{1}{2}u}-e^{-\frac{1}{2}u}} = -e^{-\frac{1}{2}u} + \frac{1}{e^{\frac{1}{2}u}-e^{-\frac{1}{2}u}} \,,
	\end{equation}
	when both sides are expanded in series about $u=0$. Ergo,
	\begin{equation}
		\left[P_2(e^{\hat{x}^{\vee}-\hbar/2})e^{\hat{y}^{\vee}+\hbar/2}-P_1(e^{\hat{x}^{\vee}-\hbar/2})\right] \psi^{\vee}(z;b) = 0 \,,
	\end{equation}
	holds and taking the Laplace transform yields the desired result.
\end{proof}

\begin{remark}\label{r:lim}
	It would be nice to drop some of the assumptions of the above theorem. Indeed, one could try dropping the condition that $P_1/P_2$ has only simple zeros and poles as follows. Write
	\begin{equation}
		P_1(e^y) = A\prod_{i=1}^{d_1} (e^y - a_i)^{m_i},\,\qquad P_2(e^y) = B\prod_{i=1}^{d_2} (e^y - b_i)^{m_i}\,. 
	\end{equation}
	Then define
	\begin{equation}
		P_1^{\bs{\epsilon}}(e^y) = A\prod_{i=1}^{d_1} (e^y - a_i + \epsilon_i)^{m_i},\,\qquad P_2^{\bs{\delta}}(e^y) = B\prod_{i=1}^{d_2} (e^y - b_i + \delta_i)^{m_i}\,.,
	\end{equation}
	where $\bs{\epsilon} = (\epsilon_1,\dots,\epsilon_{d_1}) \in \C^{d_1}$ and $\bs{\delta} = (\delta_1,\dots,\delta_{d_2}) \in \C^{d_2}$. For generic values of $\bs{\epsilon}$ and $\bs{\delta}$ the curve
	\begin{equation}
		P_{\bs{\epsilon},\bs{\delta}}(x,y) = P_2^{\bs{\delta}}(e^y)e^x - P_1^{\bs{\epsilon}}(e^y) = 0\,,
	\end{equation}
	satisfies the conditions of \cref{t:QC}. What happens in the limit $\bs{\epsilon},\bs{\delta}\to \bs{0}$? The answer is provided by \cite{BBCKS23}. The limit commutes with TR, provided one takes into account \emph{all ramification points} (see \cite{BE13,BBCKS23} for the appropriate definition of TR at non-simple ramification points), not just those that are zeros of $dx$. In particular, if $P_1/P_2$ has a non-simple zero or poles, then $x$ should be taken to be ramified at this point. If one includes contributions to the correlators $\omega_{g,n}$ from all ramification points, then the above theorem should still hold without the condition that $P_1/P_2$ has only simple zeros and poles. However, in applications, such as the relation to knot theory, these ramification points are often ignored \cite{BE12}.
\end{remark}

\begin{remark}\label{r:QC}
	One could also prove a general theorem for curves of the form $P_2(e^y)x - P_1(e^y) = 0$ using similar logic, but with the only condition being that one of the two limits
	\begin{equation}
		\lim\limits_{z\to 0} \frac{P_1(z)}{P_2(z)}\,,\qquad \lim\limits_{z\to \infty} \frac{P_1(z)}{P_2(z)}\,,
	\end{equation}
	where $z=e^y$, is infinite (this condition means one has an appropriate basepoint $b$). The result is
	\begin{equation}
		\left[P_2(e^{\hat{y}})\left(\hat{x}+\frac{1}{2}\hbar\right)-P_1(e^{\hat{y}})\right]\psi(z;b) = 0\,.
	\end{equation}
	Quantum curves of this form are computed in \cref{ss:GW}.
\end{remark}

The rest of this work will be devoted to constructing quantum curves in interesting examples.

\subsection{Examples}

\subsubsection{$P(x,y) = e^x + e^{(f+1)y} - e^{fy} = 0$}\label{sss:mc}

Here, for an integer $f\in\Z$ called the framing, the quantisation of the spectral curve
\begin{equation}\label{e:C3}
	\mc{S} = \left(\P^1,\, x(z) = f\log(z) + \log(1-z),\, y(z) = \log(z),\, \omega_{0,2}(z_1,z_2) = \frac{dz_1dz_2}{(z_1-z_2)^2}\right)\,,
\end{equation}
is computed. The correlators associated to this curve are known to calculate the Gromov-Witten invariants of framed toric Calabi-Yau threefolds when expanded around $z=1$ in the variable $e^{x(z)}$ \cite{BKMP08,BS12,Z12,FLZ20}.

Here, as long as $f\neq 0$, the conditions of \cref{t:QC} are satisfied, so the following quantum curve is obtained for the basepoint $b=0,\infty$
\begin{equation}
	\left[e^{(f+1)\hbar/2}e^{\hat{x}} + e^{-\hbar/2}e^{(f+1)\hat{y}} - e^{f\hat{y}}\right]\psi(z;b) = 0\,.
\end{equation}

A quantum curve in this case was previously obtained in \cite{Z12}, where the wavefunction killed by the quantum curve in \cite{Z12} has the natural choice of basepoint $b=1$.

\subsubsection{$P(x,y) = c^P e^x(e^{y}-c^2)^Q - c^Q e^{Py}(e^{y}-1)^Q = 0$}\label{sss:QP}

Take two integers $Q,P\in \Z$,\footnote{The integer $P$ should not be confused with the relation $P(x,y) = 0$.} $c\in\C\setminus\{0,1\}$, and examine the following spectral curve
\begin{multline}\label{e:QPTor}
	\mc{S} = \left(\P^1 ,\, x(z) = P\log(z) + Q\log(z-1) - Q\log(z-c^2) + (Q-P)\log(c),\right. \\ 
	\left. y(z) = \log(z),\, \omega_{0,2}(z_1,z_2) = \frac{dz_1dz_2}{(z_1-z_2)^2}\right)\,.
\end{multline}
The correlators of this curve are related to $(Q,P)$ torus knots\footnote{Torus knots are knots that one can view as living on the surface of an unknotted torus embedded in $\R^3$.} \cite{ABM12,BE12,GJKS14,DPSS19}.

Here, to satisfy the conditions of $\cref{t:QC}$, $P$ can not be zero and $|Q|$ must be less than or equal to one. The restriction on $Q$ means the associated Torus knot is just the unknot \cite{BE12}. By \cref{t:QC}, the following quantum curve is obtained with basepoint $b=0,\infty$
\begin{equation}
	\begin{split}
		Q\geq 0:& \qquad \left[c^Pe^{\hat{x}+\hbar/2}(e^{\hat{y}+\hbar/2}-c^2)^Q - c^Qe^{P(\hat{y}-\hbar/2)}(e^{\hat{y}-\hbar/2}-1)^Q\right] \psi(z;b) = 0\,, \\
		Q\leq 0:& \qquad \left[c^Pe^{\hat{x}+\hbar/2}(e^{\hat{y}+\hbar/2}-1)^{-Q} - c^Qe^{P(\hat{y}-\hbar/2)}(e^{\hat{y}-\hbar/2}-c^2)^{-Q}\right] \psi(z;b) = 0\,.
	\end{split}	
\end{equation}

For the spectral curve of \cref{e:QPTor}, a quantum curve was previously obtained in \cite{DPSS19}. This quantum curve is different than the one obtained here, as the basepoint $b=0,\infty$ is not the natural one from the perspective of the knot theory application.

\subsubsection{$P(x,y) = e^x - ye^{-\phi(y)} = 0$}\label{sss:H} 

Let $\phi$ be a meromorphic function on $\P^1$ and consider the spectral curve
\begin{equation}\label{e:SCHurwitz}
	\mc{S} = \left(\P^1,\, x(z) = \log z - \phi(z),\, y(z) = z,\, \omega_{0,2}(z_1,z_2) = \frac{dz_1dz_2}{(z_1-z_2)^2} \right)\,.
\end{equation}
For various choices of $\phi(z)$ the correlators $\omega_{g,n}$ produced from $\mc{S}$ can be viewed as generating functions of different sorts of Hurwitz numbers, which are counts of ramified coverings from a genus $g$ Riemann surface to $\P^1$ with given ramification behaviour (different `given ramification behaviours' lead to different sorts of Hurwitz numbers), up to equivalence, and weighted by automorphism \cite{BDKS20a}. To extract the Hurwitz numbers, one expands the correlators $\omega_{g,n}$ around $z=0$ in the variable $e^{x(z)}$ \cite{ABDKS24b}. For example, the choice $\phi(z) = z$ leads to $\omega_{g,n}$ that calculate the simple Hurwitz numbers when expanded around $z=0$ in the variable $ze^{-z}$ \cite{BM08,EMS11}. 

For the basepoint $b=\infty$ the dual wavefunction of $\mc{S}$ is easily computed
\begin{equation}
	\psi^{\vee}(z;\infty) = \exp\left(\frac{1}{\hbar}\int^z\phi dy - \frac{1}{2}\csch(\frac{1}{2}\hbar\partial_z)\log z\right)\,,
\end{equation}
where $\partial_z^{-1}$ indicates indefinite integration with respect to $z=y$. Via an analogous computation to the one performed in the proof of \cref{t:QC} one finds
\begin{equation}
	 \frac{1}{z-\hbar/2}e^{-\hbar\partial_z}e^{-\frac{1}{2}\csch(\frac{1}{2}\hbar\partial_z)\log z} = e^{-\left(e^{-(\hbar/2)\partial_z}+e^{-\hbar\partial_z}\frac{1}{2}\csch(\frac{1}{2}\hbar\partial_z)\right)\log z} = e^{-\frac{1}{2}\csch(\frac{1}{2}\hbar\partial_z)\log z}\,,
\end{equation}
so it follows that
\begin{equation}
	\left[\hat{x}^{\vee}-\frac{1}{2}\hbar - e^{\hat{y}^{\vee}}\right]e^{-\frac{1}{2}\csch(\frac{1}{2}\hbar\partial_z)\log z} = 0 \Rightarrow e^{\frac{1}{\hbar}\int^z \phi dy}\left[\hat{x}^{\vee}-\frac{1}{2}\hbar - e^{\hat{y}^{\vee}}\right]e^{-\frac{1}{\hbar}\int^z \phi dy}\psi^{\vee}(z;\infty) = 0\,.
\end{equation}
Taking the inverse Laplace transform yields the following proposition
\begin{proposition}
	The quantum spectral curve of $\mc{S}$ with basepoint $b=\infty$ is
	\begin{equation}\label{e:QCHur}
		\left[\hat{y}-\frac{1}{2}\hbar - e^{\frac{1}{\hbar}\int^{\hat{y}} \phi dy} e^{\hat{x}} e^{-\frac{1}{\hbar}\int^{\hat{y}} \phi dy}\right] \psi(z;\infty) = 0\,,
	\end{equation}
	where the notation $\int^{\hat{y}}\phi dy$ means to take an antiderivative of $\phi$, and then formally substitute $\hat{y}$ as its argument; notice the arbitrary constant in the antiderivative cancels in the above formula.
\end{proposition}

\begin{remark}
	By the Baker-Campbell-Hausdorff formula and the fact that $[\hat{x},f(\hat{y})] = -\hbar f'(\hat{y})$
	\begin{equation}
		e^{\frac{1}{\hbar}\int^{\hat{y}} \phi dy} e^{\hat{x}} e^{-\frac{1}{\hbar}\int^{\hat{y}} \phi dy} = e^{\phi(\hat{y}) + \mathcal{O}(\hbar)}\,,
	\end{equation}
	so it is indeed correct to call \cref{e:QCHur} a quantum curve.
\end{remark}

\begin{remark}
	Often in the literature the function $X(z) = e^{x(z)}$ is used in the definition of the spectral curve $\mc{S}$ and the definition of the $1$-form $\omega_{0,1}$ is modified so that $\omega_{0,1} \coloneq ydX/X$. This changes the $\hbar^0$ term in the $\hbar$ expansion of the logarithm of the wavefunction
	\begin{equation}
		-\frac{1}{2}\log\frac{dx}{dz} = \frac{1}{2}x - \frac{1}{2}\log\frac{dX}{dz}\,.
	\end{equation}
	Therefore, to change to the convention of using $X$, one should multiply the wavefunction by $e^{-\frac{1}{2}x}$. This changes the quantum curve via the shift $\hat{y} \to \hat{y}+\frac{1}{2}\hbar$, so rather than \cref{e:QCHur} one instead obtains
	\begin{equation}
		\left[\hat{y} - e^{\frac{1}{\hbar}\int^{\hat{y}+\frac{1}{2}\hbar} \phi dy} \hat{X} e^{-\frac{1}{\hbar}\int^{\hat{y}+\frac{1}{2}\hbar} \phi dy}\right] \psi(z;\infty) = 0\,,
	\end{equation}
	where $\hat{y} \coloneq \hbar X d/dX$.
\end{remark}

Unfortunately, the above wavefunction is constructed for the basepoint $b=\infty$, whereas one wants to expand around the point $z=0$ to obtain the Hurwitz numbers; therefore the above wavefunction does not have a natural enumerative interpretation. Quantum curves with the natural choice of basepoint $b=0$ have been obtained previously for special cases of spectral curves of the form \eqref{e:SCHurwitz} \cite{MSS13,ALS16}.

\subsection{Gromov-Witten theory of $\P(a,b)$}\label{ss:GW}

\subsubsection{Connected stationary Gromov-Witten invariants of $\P^1$}\label{sss:GWP1}

Here the spectral curve
\begin{equation}\label{e:GW}
	\mc{S} = \left(\P^1 ,\, x(z) = z+z^{-1} ,\, y(z) = \log(z) ,\, \omega_{0,2}(z_1,z_2) = \frac{dz_1dz_2}{(z_1-z_2)^2}\right)\,,
\end{equation}
is examined. The functions $x$ and $y$ satisfy identically the equation $P(x,y) = x - 2\cosh(y) = 0$. This curve is of interest as the corresponding correlators are generating functions for the Gromov-Witten invariants of $\P^1$ \cite{NS14,DOSS14,FLZ17}.

To be more precise let $\overline{\mc{M}}_{g,n}(\P^1,d)$ be the moduli space of degree $d$ stable maps from an algebraic curve of genus $g$ with $n$ marked points to the complex projective line. For $i=1,\dots,n$ denote by $\mathrm{ev}_i:\overline{\mc{M}}_{g,n}(\P^1,d)\to \P^1$ the map defined by evaluating a stable map at the $i^\text{th}$ marked point and take $\mc{L}_i$ to be the cotangent bundle over the $i^\text{th}$ marked point. Then let $\psi_i \in H^2(\overline{\mc{M}}_{g,n}(\P^1,d);\Q)$ be the first Chern class of $\mc{L}_i$. Finally, for $a_1,\dots,a_n\in \Z_{\geq 0}$ define the connected stationary Gromov-Witten invariants as
\begin{equation}
	\left\<\prod_{i=1}^n \tau_{a_i}(\omega) \right\>^d_{g,n} \coloneq \delta_{\sum_i a_i,\,2g-2+2d} \int_{\left[\overline{\mc{M}}_{g,n}(\P^1,d)\right]^{\text{vir}}} \prod_{i=1}^n\psi_i^{a_i}\mathrm{ev}_i^*(\omega)\,,
\end{equation}
where $\left[\overline{\mc{M}}_{g,n}(\P^1,d)\right]^{\text{vir}}$ is the virtual fundamental class of $\overline{\mc{M}}_{g,n}(\P^1,d)$ and $\omega \in H^2(\P^1;\Q)$ is Poincar\'e dual to the point class. For more details on the definition of Gromov-Witten invariants see \cite[section 3]{NS14}.

It is shown in \cite[theorem 5.2]{DOSS14} (the $g=0,1$ case was established in \cite[theorem 1.1]{NS14}), that for $2g+n-2>0$
\begin{equation}\label{e:GWP1}
	\omega_{g,n}(z_{\llbracket n \rrbracket}) \stackrel{x(z_i)\to\infty}{\sim} (-1)^n\sum_{a_1,\dots,a_n=0}^{\infty} \left\<\prod_{i=1}^n \tau_{a_i}(\omega) \right\>_{g,n} \prod_{i=1}^n (a_i+1)! x(z_i)^{-a_i-2} dx(z_i) \,,
\end{equation}
where the omission of the superscript $d$ in the Gromov-Witten invariant indicates that $d$ is chosen such that $\sum_i a_i = 2g-2+2d$ holds. The unstable cases $(g,n) = (0,1),(0,2)$ have the following expansions \cite{NS14} 
\begin{equation}
	\begin{split}
		\omega_{0,1}(z_1) &\stackrel{x(z_1)\to\infty}{\sim} \log x(z_1) dx - \sum_{a_1=0}^{\infty} \left\< \tau_{a_1}(\omega) \right\>_{g,n} \frac{(a_1+1)!}{x(z_1)^{a_1+2}} dx(z_1)\,,\\
		\omega_{0,2}(z_1,z_2) &\stackrel{x(z_1),x(z_2)\to\infty}{\sim} \frac{dx(z_1)dx(z_2)}{(x(z_1)-x(z_2))^2} + \sum_{a_1,a_2=0}^{\infty} \left\< \tau_{a_1}(\omega)\tau_{a_2}(\omega) \right\>_{g,n} \frac{(a_1+1)!(a_2+1)!}{ x(z_1)^{a_1+2} x(z_1)^{a_2+2} } dx(z_1) dx(z_2)\,.
	\end{split}
\end{equation}

Thus the wavefunction $\psi(z;0)$ takes the form
\begin{equation}
	\psi(z;0) = \exp\left(\frac{1}{\hbar}x(z)(\log x(z) - 1) + \sum_{n=1}^{\infty}\sum_{g=0}^{\infty}\frac{\hbar^{2g+n-2}}{n!} \left\< \left(\sum_{a=0}^{\infty}\frac{a!\tau_{a}(\omega)}{x(z)^{a+1}}\right)^n \right\>_{g,n}\right)\,.
\end{equation}

Notice that $dx$ and $dy$ both have a pole at $b=0$, so the above wavefunction can be computed as the Laplace transform of the dual wavefunction by \cref{t:Lap}. As all the dual stable correlators vanish, the dual wavefunction itself is easily computed as
\begin{equation}
	\psi^{\vee}(z; 0) = \exp\left(\frac{1}{2}y(z) - \int^z x dy \right)\,,
\end{equation}
which means
\begin{equation}
	\hat{y}^{\vee} \psi^{\vee}(z; 0) = \left(x(e^{y(z)}) - \frac{1}{2}\hbar\right) \psi^{\vee}(z;0)\,,
\end{equation}
recalling that $\hat{y}^{\vee} = -\hbar d/dy$. Thus, taking the Laplace transform of the quantisation of the dual spectral curve one immediately obtains the following proposition.
\begin{proposition}\label{p:GWcurve}
	The quantisation of the spectral curve of \cref{e:GW} with basepoint $b=0$ is given by
	\begin{equation}
		\left[\hat{x} + \frac{1}{2}\hbar - 2\cosh(\hat{y})\right]\psi(z; 0) = 0 \,.
	\end{equation}
\end{proposition}
A quantisation with an alternative integration scheme has been computed previously \cite{DMNPS17}. The quantisation derived here is actually a special case of \cite[theorem 2.2]{CG19} (it should be noted, however, that \cite{CG19} did not obtain their result from the topological recursion, and the methodology used in \cite{DMNPS17} was specific to this situation rather than a general method to obtain quantisations). 

\subsubsection{$T$-equivariant Gromov-Witten theory of $\P(a,b)$}\label{sss:GWPab}

For $a,b\in\Z_{\geq 1}$ coprime\footnote{Although the geometrical context requires that $a$ and $b$ are coprime, \cref{t:GW} will still give a quantisation of the spectral curve \eqref{e:GWT} when $a$ and $b$ are not coprime.} define the weighted projective line $\P(a,b) \coloneq \left(\C^2\setminus\{(0,0)\}\right)/\sim$ where $(p_1,p_2)\sim (\lambda^b p_1, \lambda^a p_2)$ for $\lambda\in\C^*$, and denote by $[p_1:p_2]$ a homogenous coordinate on $\P(a,b)$. Let the Torus $\mb{T}\coloneq (\C^*)^2$ act on $\P(a,b)$ via $(t_1,t_2)\cdot [p_1:p_2] = [t_1p_1 : t_2p_2]$. 

Consider now the spectral curve
\begin{equation}\label{e:GWT}
	\mathcal{S} = \left(\P^1,\, x(z) = \sum_{l=-b}^{a}\left[q_lz^l + w_l\log(q_lz^l)\right],\, y(z) = \log z,\, \omega_{0,2}(z_1,z_2) = \frac{dz_1dz_2}{(z_1-z_2)^2}\right)\,,
\end{equation}
where $q_l,w_l\in \C$ such that $q_a,q_{-b}\neq 0$ and $q_l=0\Rightarrow w_l=0$; one may assume $q_a = 1$ by rescaling the affine coordinate $z$ and $q_0=w_0=0$ as adding a constant to $x$ does not affect the topological recursion. This spectral curve is related to the $T$-equivariant Gromov-Witten invariants of $\P(a,b)$ as sketched below.

As discussed, assume $q_a=1$ and $w_0=0$. Let $T \coloneq \mb{T}^{a+b}$ and denote by $H^*_{T}(\text{pt};\C) = \C[\bs{w}] \coloneq \C[w_{-b},\dots,w_{-1},w_{1},\dots,w_{a}]$ the $T$-equivariant cohomology of a point. Then the $T$-equivariant cohomology of $\P(a,b)$ is given by \cite{T18}
\begin{equation}
	H^*_{T}(\P(a,b);\C) = \C[\bs{w}][X,\bar{X},\bs{p}] / \left< X\bar{X},aX^a-b\bar{X}^b+\bs{p} \right> \,,\quad \bs{p} \coloneq \sum_{l=-b}^a lw_l \,,
\end{equation}
where $\deg X = \deg \bar{X} = \deg w_i = 2$. 

Analogously to the (non-equivariant) $\P^1$ case let $\overline{\mc{M}}_{g,n}(\P(a,b),d)$ be the moduli space of degree $d$ stable maps from an algebraic curve of genus $g$ with $n$ marked points to $\P(a,b)$. Let $\mathrm{ev}_i:\overline{\mc{M}}_{g,n}(\P(a,b),d)\to \P(a,b)$ be the natural evaluation maps, take $\mc{L}_i$ to be the cotangent bundle over the $i^\text{th}$ marked point and put $\psi_i \in H^2(\overline{\mc{M}}_{g,n}(\P(a,b),d);\Q)$ as $\mc{L}_i$'s first Chern class. For $a_1,\dots,a_n\in \Z_{\geq 0}$ and $\gamma_1,\dots,\gamma_n \in H^*_T(\P(a,b);\C)$ define the orbifold descendant Gromov-Witten invariants as
\begin{equation}
	\left\<\prod_{i=1}^n \tau_{a_i}(\gamma_i) \right\>^{d,T}_{g,n} \coloneq \int_{\left[\overline{\mc{M}}_{g,n}(\P(a,b),d)\right]^{\text{vir}}} \prod_{i=1}^n\psi_i^{a_i}\mathrm{ev}_i^*(\gamma_i)\,,
\end{equation}
where $\left[\overline{\mc{M}}_{g,n}(\P(a,b),d)\right]^{\text{vir}}$ is the virtual fundamental class of $\overline{\mc{M}}_{g,n}(\P(a,b),d)$, and it is assumed the Gromov-Witten invariant is zero if the degree of the cohomology class does not align with the virtual dimension.

Let $T_i = X^i\in H^*_T(\P(a,b);\C)$. For $i\neq 0$ put $q_i = Qe^{t_{-i}}$ for some $Q,t_i\in \C$ and let $t_0 = q_0$, so there are $a+b$ variables $t_i$ (as $q_a$ was set to one, there is no $t_a$). Let $\bs{t} \coloneq \sum_i t_{i}T_i$ and for $\gamma_1,\dots,\gamma_n \in H^*_T(\P(a,b);\C)$ define
\begin{equation}
	\left<\left\< \prod_{i=1}^n \tau_{a_i}(\gamma_i) \right\>\right>^{T}_{g,n} \coloneq \sum_{d=0}^{\infty}\sum_{m=0}^{\infty} \frac{Q^d}{m!} \left\< \left(\prod_{i=1}^n \tau_{a_i}(\gamma_i)\right) \big(\tau_0(\bs{t})\big)^m \right\>^{d,T}_{g,n+m}\,.
\end{equation}
For $i=1,\dots,n$ let $\bs{u}_i = \sum_{a\in \Z_{\geq 0}} (u_j)_a z^a$ where $(u_j)_a \in H^*_T(\P(a,b);\C)$ and define
\begin{equation}
	F_{g,n}^{T}(\bs{u}_1,\dots,\bs{u}_n,\bs{t}) = \sum_{a_1,\dots,a_n\in \Z_{\geq 0}} \left<\left\< \prod_{i=1}^n \tau_{a_i}((u_i)_{a_i}) \right\>\right>^{T}_{g,n}\,.
\end{equation}
With an appropriate identification of the $\bs{u}_i$ these $F_{g,n}^T$ are essentially the $\omega_{g,n}$ obtained from the spectral curve $\mc{S}$ \cite[theorem 1]{T18} (see also \cite[theorem 1]{FLZ17}). In the case when $a=b=1$, $w_l=q_0=0$, and $q_1=q_{-1}=1$ it was shown in \cite[theorem 4.2]{FLZ17} that this result reduces to what is stated in \cref{e:GWP1}.

However, despite the fact that the spectral curve of \cref{e:GWT} looks significantly more complicated than the one of \cref{e:GW}, as $dy$ has a pole at both $z=0,\infty$ (the only logarithmic singularities of $x$) the stable dual correlators are trivial and the dual wavefunction takes the same simple form
\begin{equation}
	\psi^{\vee}(z; 0) = \exp\left(\frac{1}{2}y(z) - \int^z x dy \right)\,,
\end{equation}
so one can still conclude
\begin{equation}\label{e:dualWGW}
	\hat{y}^{\vee} \psi^{\vee}(z; 0) = \left(x(e^{y(z)}) - \frac{1}{2}\hbar\right) \psi^{\vee}(z; 0)\,.
\end{equation}
Taking the Laplace transform will yield the following theorem.
\begin{theorem}\label{t:GW}
	The quantum spectral curve of $\mc{S}$ with basepoint $b=0$ is given by
	\begin{equation}
		\left[ \hat{x}+\frac{1}{2}\hbar - \sum_{l=-b}^{a} \left[ q_le^{l\hat{y}}+w_l(l\hat{y}+\log q_l) \right] - C \right]\psi(z;0) = 0\,,
	\end{equation}
	where the constant $C\in\bigoplus_{l}2w_l\pi i\Z$ depends on the branch choice of the logarithm.
\end{theorem}
\begin{proof}
	One starts from \cref{e:dualWGW} and writes $x$ as 
	\begin{equation}
		x(e^{y}) = C + \sum_{l=-b}^{a} \left[ q_le^{ly}+w_l(ly+\log q_l) \right]\,,
	\end{equation}
	where the constant $C\in\bigoplus_{l}2w_l\pi i\Z$ is introduced to account for the fact $w_l\log(q_le^{ly})$ and $w_l(\log(q_l) + ly)$ may differ by an element of $2w_l\pi i\Z$ depending on the branch choice of the logarithm. Taking the Laplace transform of \cref{e:dualWGW} with this rewritten $x$ immediately yields the desired result.
\end{proof}

A special case of this result with $a=q_{1}=1$ and $w_l=q_{1-b}=\cdots=q_{0}=0$ was found in \cite[theorem 2.2]{CG19},\footnote{When comparing with \cite{CG19} one should substitute $\hbar \to -\hbar$, $b\to r$, and $q_b \to q^{r}$. Also, the result here is for the so-called principle specialisation, which corresponds to $t=0$ in \cite{CG19}.} where the result was obtained directly from the study of the Gromov-Witten invariants rather than through the topological recursion. A special case of \cite[equation (4.26)]{A21} establishes the above theorem when $a=1$ and $w_l=0$; \cite{A21} obtained his result through the direct study of corresponding matrix models rather than through the topological recursion. 


\section{Conclusion and outlook}\label{s:con}

In this paper the notion that the wavefunction associated to a spectral curve and the dual wavefunction associated to the dual spectral curve are Laplace transforms of one another has been made rigorous. This definition of the Laplace transform acts as a mapping between spaces of formal series in $\hbar$, while still satisfying the key properties one would expect of the Laplace transform.

The implications for the forms of quantisations of the original spectral curve were then examined. In particular, it was shown that, given a quantum curve for the dual spectral curve one could very simply derive a quantum curve for the original spectral curve. This insight was put to use to derive quantum curves for many spectral curves that have particularly simple duals. A general theorem for the quantisation of a large class of curves of the form $P(e^x,e^y)=0$ was formulated, and several quantisations were worked out explicitly, including a quantisation for the spectral curve associated to the $T$-equivariant Gromov-Witten theory of $\P(a,b)$.

Nevertheless, a few interesting questions remain.
\begin{itemize}
	\item All results of the present work are in genus zero. However, many curves that one would like to quantise, especially those of the form $P(e^x,e^y) = 0$, are of higher genus \cite{EGMO21}. Is there a higher genus generalisation of \cref{t:Lap} (see also: \cref{r:Highg})?
	\item The theorem that allows one to quantise genus zero curves of the form $P(e^x,e^y)=0$ has some rather technical restrictions. In what ways are these restrictions necessary, and in what ways can they be lifted (see also: \cref{r:lim})?
	\item The present work only examined the case when the basepoint for the wavefunction and its dual are the same, and both $dx$ and $dy$ posses a pole at this point. Is there a way to interpret the Laplace transform of a general wavefunction (i.e., defined with an arbitrary integration scheme) as a dual wavefunction?
\end{itemize}

{\setlength\emergencystretch{\hsize}\hbadness=10000
\printbibliography}

\end{document}